\documentclass[proceedings]{stacs}
\stacsheading{2010}{621-632}{Nancy, France}
\firstpageno{621}

\usepackage{textcomp}
\usepackage{amsmath,amssymb,amsthm,mathdots}
\usepackage{graphicx}
\usepackage{psfrag}

\newcommand{\E}{\ensuremath{\mathbb{E}}}

\newcommand{\occ}{\ensuremath{{\rm occ}}}
\newcommand{\dist}{\ensuremath{{\rm dist}}}

\newcommand{\flin}{\ensuremath{f_{\rm LIN}}}
\newcommand{\focc}{\ensuremath{f_{\rm occ}}}

\newcommand{\var}{\ensuremath{{\rm vbl}}}

\newcommand{\NNnull}{\ensuremath{\mathbb{N}_0}}

\newcommand{\ignore}[1]{}

\begin{document}
 \title{Unsatisfiable Linear CNF Formulas Are Large and 
 Complex}
 \author{D. Scheder}{Dominik Scheder}
 \address{ETH Z\"urich, Institute of Theoretical Computer Science\newline
 Universit\"atstrasse 6, CH-8092 Z\"urich, Switzerland}
 \email{dscheder@inf.ethz.ch}
\thanks{Research is supported by the SNF Grant 200021-118001/1}
 \subjclass{Computational and structural complexity}


\begin{abstract}
  We call a CNF formula {\em linear} if any two clauses have at most
  one variable in common. We show that there exist unsatisfiable
  linear $k$-CNF formulas with at most $4k^24^k$ clauses, and on the
  other hand, any linear $k$-CNF formula with at most
  $\frac{4^k}{8e^2k^2}$ clauses is satisfiable.  The upper bound uses
  probabilistic means, and we have no explicit construction coming
  even close to it. One reason for this is that unsatisfiable linear
  formulas exhibit a more complex structure than general (non-linear)
  formulas: First, any treelike resolution refutation of any
  unsatisfiable linear $k$-CNF formula has size at least
  $2^{2^{\frac{k}{2}-1}}$. This implies that small unsatisfiable
  linear $k$-CNF formulas are hard instances for Davis-Putnam style
  splitting algorithms. Second, if we require that the formula $F$ have
  a {\em strict} resolution tree, i.e. every clause of $F$ is used
  only once in the resolution tree, then we need at least
  $a^{a^{\iddots^a}}$ clauses, where $a \approx 2$ and the height of
  this tower is roughly $k$.
\end{abstract}

\maketitle

\section{Introduction}

How can CNF formulas become unsatisfiable? Roughly speaking, there are
two ways: Either some constraint (clause) is itself impossible to
satisfy -- the empty clause; or, every clause can be satisfied
individually, but one cannot satisfy all of them simultaneously.  In
the latter case, the clauses have to somehow overlap.  How much?  For
example, take $k$ boolean variables $x_1,\dots,x_k$.  The conjunction
of all $2^k$ possible clauses of size $k$ is the {\em complete $k$-CNF
  formula} and denote by $\mathcal{K}_k$. It is unsatisfiable, and as
small as possible: Any $k$-CNF formula with less than $2^k$ clauses is
satisfiable. Clearly, the clauses of $\mathcal{K}_k$ overlap a lot.
What if we require that any two distinct clauses share at most one
variable? We call such a formula {\em linear}. There are unsatisfiable
linear $k$-CNF formulas, but they are significantly larger and have a
much more complex structure
than $\mathcal{K}_k$.\\

A CNF formula is a conjunction (AND) of {\em clauses}, and a clause is
a disjunction (OR) of {\em literals}. A literal is either a boolean
variable $x$ or its negation $\bar{x}$. We require that a clause does
not contain the same literal twice, and does not contain complementary
literals, i.e., both $x$ and $\bar{x}$. To simplify notation, we also
regard formulas as sets of clauses and clauses as sets of literals. A
clause with $k$ literals is a $k$-clause, and a $k$-CNF formula is a
CNF formula consisting of $k$-clauses. For a clause $C$, we denote by
$\var(C)$ set of variables $x$ with $x \in C$ or $\bar{x} \in C$.
Consequently, a CNF formula $F$ is {\em linear} if $|\var(C) \cap
\var(D)| \leq 1$ for any two distinct clauses $C,D \in F$. As a
relaxation of this notion, we call $F$ {\em weakly linear} if $|C \cap
D| \leq 1$ for any distinct $C,D \in F$.\\

{\em Example.} The formula $(\bar{x}_1 \vee x_2) \wedge (\bar{x}_2
\vee x_3) \wedge (x_3 \vee x_4) \wedge (\bar{x}_4 \vee \bar{x}_1) $ is
linear, whereas $(\bar{x}_1 \vee x_2) \wedge (x_1 \vee x_2) \wedge
(x_2 \vee x_3)$ is weakly linear, but not linear, and finally
$(x_1\vee x_2 \vee x_3) \wedge (x_1 \vee x_2 \vee \bar{x}_3)$ is not
weakly linear (and not linear, either).\\

It is not very difficult to construct an unsatisfiable linear $2$-CNF
formula, but significantly more effort is needed for a $3$-CNF
formula. It is not obvious whether unsatisfiable linear $k$-CNF
formulas exist for every $k$. These questions have been asked first by
Porschen, Speckenmeyer and Randerath~\cite{PRS06}, who also proved
that for any $k\geq 3$, if an unsatisfiable linear $k$-CNF formula
exists, then deciding satisfiability of linear $k$-CNF formulas is
NP-complete. Later, Porschen, Speckenmeyer and Zhao~\cite{PSZ08} and,
independently, myself~\cite{Scheder07} gave a construction of
unsatisfiable linear $k$-CNF formulas, for every $k \in \NNnull$:
\begin{theorem}[\cite{PSZ08}, \cite{Scheder07}]
  \label{theorem-existence}
  For every $k \geq 0$, there exists an unsatisfiable linear $k$-CNF
  formula $F_k$, with $F_0$ containing one clause and $F_{k+1}$
  containing $|F_k| 2^{|F_k|}$ clauses.
\end{theorem}
The $|F_k|$ are extremely large. Here, we will
give an almost optimal  construction.
\ignore{
These formulas are extremely large. Printing $F_4$ would exceed the
amount of paper available in the universe. In this paper, we determine
the size of unsatisfiable linear $k$-CNF formulas rather precisely:}
\begin{theorem}
  All weakly linear $k$-CNF formulas with at most
  $\frac{4^k}{8e^2(k-1)^2}$ clauses are satisfiable. There exists an
  unsatisfiable linear $k$-CNF formula with $4k^24^k$ clauses.
  \label{theorem-bounds}
\end{theorem}
It is a common phenomenon in extremal combinatorics that by
probabilistic means one can show that a certain object exists (in our
case, a ``small'' linear unsatisfiable $k$-CNF formula), but one
cannot explicitly construct it. We have no explicit construction
avoiding the tower-like growth in Theorem~\ref{theorem-existence}. We
give some arguments why this is so, and show that small linear
unsatisfiable $k$-CNF formulas have a more complex structure than
their non-linear relatives. To do so, we speak about {\em resolution}.

\subsection{Resolution Trees} 

If $C$ and $D$ are clauses and there is unique literal $u$ such that
$u\in C$ and $\bar{u} \in D$, then $(C \setminus \{u\}) \cup (D
\setminus \{\bar{u}\})$ is called the {\em resolvent} of $C$ and $D$.
It is easy to check that every assignment satisfying $C$ and $D$ also
satisfies the resolvent.

\begin{definition}
  A {\em resolution tree} for a CNF formula $F$ is a tree $T$ whose
  vertices are labeled with clauses, such that
  \begin{itemize}
  \item each leaf of $T$ is labeled with a clause of $F$,
  \item the root of $T$ is labeled with the empty clause,
  \item if vertex $a$ has children $b$ and $c$, and these are
   labeled with clauses $C_a, C_b, C_c$, respectively, then $C_a$
    is the resolvent of $C_b$ and $C_c$.
  \end{itemize}
\end{definition}

It is well-known that a CNF formula $F$ is unsatisfiable if and only
if it has a resolution tree (which can be exponentially
large in $|F|$). Proving lower bounds on the size of
resolution trees (and general resolution proofs, which we will not
introduce here) has been and still is an area of intensive research.
See for example Ben-Sasson and Wigderson~\cite{BW01}.

\begin{theorem}
  Let $k \geq 2$. Every resolution tree of an unsatisfiable weakly
  linear $k$-CNF formula has at least $2^{2^{\frac{k}{2}-1}}$ leaves.
\label{theorem-treelike}
\end{theorem}

A large ratio between the size of $F$ and the size of a smallest
resolution tree is an indication that $F$ has a complex structure.
For example, it is well-known that the running time of so-called
Davis-Putnam procedures on a formula $F$ is lower bounded by the size
of the smallest resolution tree of $F$ (actually those procedures were
introduces by Davis, Logeman and Loveland~\cite{dll}). Such a
procedure tries to find a satisfying assignment for a formula $F$ (or
to prove that none exists) by choosing a variable $x$, and then
recursing on the formulas $F^{[x \mapsto 0]}$ and $F^{[x \mapsto 1]}$,
obtained from $F$ by fixing the value of $x$ to $0$ or $1$,
respectively. If $F$ is unsatisfiable, the procedure implicitly
constructs a resolution tree. \\

 A CNF formula $F$ is {\em minimal
  unsatisfiable} if it is unsatisfiable, and for every clause $C \in
F$, $F\setminus \{C\}$ is satisfiable. The complete $k$-CNF formula
introduced above is minimal unsatisfiable, and has a resolution tree
with $2^k$ leaves, one for every clause.  This is as small as
possible, since for a minimal unsatisfiable formula, every clause must
appear as label of at least one leaf of any resolution tree.  We call
a resolution tree {\em strict} if no two leaves are labeled by the
same clause, and a formula $F$ {\em strictly treelike} if it has a
strict resolution tree. In some sense, strictly treelike formulas are
the least complex formulas possible.  For example, the complete
formula $\mathcal{K}_k$ and the formulas constructed in the proof of
Theorem~\ref{theorem-existence} are strictly treelike.

\begin{theorem}
  For any $\epsilon > 0$, there exists a constant $c$ such that for
  any $k \in \mathbb{N}$, any strictly treelike weakly linear $k$-CNF
  formula has at least ${\rm tower}_{2-\epsilon} (k - c)$ clauses,
  where ${\rm tower}_a(n)$ is defined by ${\rm tower}_a(0) = 1$ and
  ${\rm tower}_a(n+1) = a^{ {\rm tower}_a(n)}$.
\label{theorem-strictly-treelike}
\end{theorem}

Strictly treelike formulas appear in other contexts, too. Consider
MU(1), the class of minimal unsatisfiable formulas whose number of
variables is one less than the number of clauses. A result
of Davydov, Davydova and Kleine B\"uning (\cite{DDKB98}, Theorem 12)
implies that every MU(1)-formula is strictly treelike. Also,
MU(1)-formulas serve as ``universal patterns'' for unsatisfiable
formulas: Szeider~\cite{Szeider2003} shows that a formula $F$ is
unsatisfiable if and only if it can be obtained from a MU(1)-formula
$G$ by renaming the variables of $G$ (in a possibly non-injective
manner). It is not difficult to show that a strictly treelike linear
$k$-CNF formula can be transformed into a linear MU(1)-formula with
the same number of clauses.

\subsection{Related Work}

For a CNF formula $F$ and a variable $x$, let $d_F(x)$ denote the {\em
  degree of $x$}, i.e. the number of clauses of $F$ containing $x$ or
$\bar{x}$, and let $d(F):=\max_x d_F(x)$ denote the {\em maximum
  degree} of $F$. For the complete $k$-CNF formula $\mathcal{K}_k$, we
have $d(\mathcal{K}_k) = 2^{k}$.  Intuitively, in an unsatisfiable
$k$-CNF formula, some variables should occur in many clauses. In other
words, the following function should be large:
\begin{eqnarray}
  f(k)  :=  \max\{ d \ \big| \ \textnormal{every k}\textnormal{-CNF
    formula }F\textnormal{ with } d(F) \leq d \textnormal { 
    is satisfiable} \} \ .
\label{eqn-f(k)}
\end{eqnarray}
The function $f(k)$ has first been investigated by
Tovey~\cite{Tovey84}, who showed $f(k) \geq k$,using Hall's Theorem.
Using the famous Lov\'asz Local Lemma (see~\cite{EL75} for the
original proof, or~\cite{AS00} for several generalized versions),
Kratochv\'{i}l, Savick\'{y} and Tuza~\cite{KST93} proved that $f(k)
\geq \frac{2^k}{ek}$, and that while all $k$-CNF formulas $F$ with
$d(F) \leq f(k)$ are trivially satisfiable, deciding satisfiability of
$k$-CNF formulas $F$ with $d(F) \leq f(k)+1$ is already NP-complete,
for $k \geq 3$. For $k=3$, this is already observed in~\cite{Tovey84}.
For an upper bound, the complete $k$-CNF formula witnesses that $f(k)
\leq 2^k-1$. Savick\'{y} and Sgall~\cite{SS2000} showed $f(k) \in
O(k^{-0.26}2^k)$.  This was improved by Hoory and Szeider~\cite{HS06}
to $f(k) \in O\left(\frac{\ln(k)2^k}{k}\right)$, and recently
Gebauer~\cite{Gebauer09} proved that $f(k) \leq \frac{2^{k+2}}{k}$.
Thus, $f(k)$ is known up to a constant factor.  The best upper bounds
on $f(k)$ come from MU(1)-formulas. This is true for large values of
$k$, since the formulas constructed in~\cite{Gebauer09} are MU(1), as
for small values: Hoory and Szeider~\cite{szeider05} show that the
function $f(k)$, when restricted to MU(1)-formulas, is computable (in
general this is not known), and derive the currently best-known bounds
on $f(k)$ for small $k$ ($k \leq 9$).  To summarize: When we try to
find unsatisfiable $k$-CNF formulas minimizing a certain
parameter, like number of clauses or maximum degree, strictly treelike
formulas do an excellent job. However, if we try to construct a small
unsatisfiable linear $k$-CNF formula, they perform horribly.  Just
compare our upper bound in Theorem~\ref{theorem-bounds} with the lower
bound for strictly
treelike formulas in Theorem~\ref{theorem-strictly-treelike}\\

While interest in linear CNF formulas is rather young, {\em linear
  hypergraphs} have been studied for quite some time. A hypergraph
$H=(V,E)$ is linear if $|e \cap f| \leq 1$ for any two distinct
hyperedges $e,f \in E$. A $k$-uniform hypergraph is a hypergraph where every
hyperedge has cardinality $k$. We ask when a hypergraph $2$-colorable,
i.e., admits a $2$-coloring of its vertices such that no hyperedge
becomes monochromatic. Bounds on the
number of edges in such a hypergraph were given by Erd\H{o}s and
Lov\'asz~\cite{EL75} (interestingly, this is the paper where the Local
Lemma has been proven). They show that there are non-$2$-colorable
linear $k$-uniform hypergraphs with $ck^44^k$ hyperedges, but not with less
than $\frac{c'4^k}{k^3}$.  The proof of the lower bound directly
translates into our lower bound for linear $k$-CNF formulas. For the
number of edges in linear $k$-uniform hypergraphs that are not $2$-colorable,
the currently best upper bound is $ck^24^k$ by Kostochka and
R\"odl~\cite{KR2009}, and the best lower bound is $k^{-\epsilon}4^k$,
for any $\epsilon > 0$ and sufficiently large $k$, due to Kostochka
and Kumbhat~\cite{KK2008}.

\section{Existence and Upper and Lower Bounds}
\label{section-upper-lower}

\begin{proof}[Proof of Theorem~\ref{theorem-existence}]
  Choose $F_0$ to be the formula consisting of only the empty clause.
  Suppose we have constructed $F_k$, and want to construct $F_{k+1}$.
  Let $m = |F_k|$. We create $m$ new variables $x_1, \dots, x_{m}$,
  and let $\mathcal{K}_m = \{D_1, D_2, \dots, D_{2^{m}}\}$ be the
  complete $m$-CNF formula over $x_1,\dots, x_{m}$.  It is
  unsatisfiable, but not linear. We take $2^{m}$ variable disjoint
  copies of $F_k$, denoted by $F^{(1)}_k,
  F^{(2)}_k,\dots,F^{(2^{m})}_k$. For each $1 \leq i \leq 2^{m}$, we
  build a linear $(k+1)$-CNF formula $\tilde{F}^{(i)}_k$ from
  $F^{(i)}_k$ by adding, for each $1 \leq j \leq m$, the
  $j$\textsuperscript{th} literal of $D_i$ to the
  $j$\textsuperscript{th} clause of $F^{(i)}_k$. Note that every
  assignment satisfying $\tilde{F}^{(i)}_k$ also satisfies $D_i$.
  Finally, we set $F_{k+1} := \bigcup_{i=1}^{2^{m}}
  \tilde{F}^{(i)}_k$.  This is an unsatisfiable linear $(k+1)$-CNF
  formula with $m2^{m}$ clauses.
\end{proof}

Using induction, it is not difficult to see that the formulas $F_k$
are strictly treelike. We will prove the upper bound in
Theorem~\ref{theorem-bounds} by giving a probabilistic construction of
a comparably small unsatisfiable linear $k$-CNF formula. Our
construction consists of two steps. First, we construct a linear
$k$-uniform hypergraph $H$ that is ``dense'' in the sense that
$\frac{m}{n}$ is large, where $m$ and $n$ are the number of hyperedges
and vertices, respectively, and then transform it randomly into a
linear $k$-CNF formula $F$ that is unsatisfiable with high
probability.

\begin{lemma}
  If there is a linear $k$-uniform hypergraph $H$ with $n$ vertices
  and $m$ edges such that $\frac{m}{n} \geq 2^k$, then there is
  an unsatisfiable linear $k$-CNF formula with $m$ clauses.
\label{random-formula}
\end{lemma}

\begin{proof}
  Let $H=(V,E)$. By viewing $V$ as a set of variables and $E$ as a set
  of clauses (each containing only positive literals), this is a
  (satisfiable) linear $k$-CNF formula.  We replace each literal in
  each clause by its complement with probability $\frac{1}{2}$,
  independently in each clause. Let $F$ denote the resulting (random)
  formula. For any fixed truth assignment $\alpha$, it holds that
  $\Pr[\alpha \textnormal{ satisfies } F] = (1-2^{-k})^m$. Hence the
  expected number of satisfying assignments of $F$ is
  $$
  2^n (1-2^{-k})^m < 2^n e^{-2^{-k}m} = e^{\ln(2)n - 2^{-k}m} \leq 1 \ ,
  $$
  where the last inequality follows from $\frac{m}{n} \geq 2^k$.
  Hence some formula $F$ has fewer than one satisfying assignment,
  i.e., none.
\end{proof}

How can we construct a dense linear hypergraph? We use a construction
by Kuzjurin~\cite{Kuzjurin1995}. Our application of this construction
is motivated by Kostochka and R\"odl~\cite{KR2009}, who use it to
construct linear hypergraphs of large chromatic number.

\begin{lemma}
  For any prime power $q$ and any $k \in \mathbb{N}$, there exists
  a $k$-uniform linear hypergraph with $kq$ vertices and $q^2$ edges.
\label{lemma-vdm}
\end{lemma}

With $n=kq$, this hypergraph has $n^2 / k^2$ hyperedges.  This is
almost optimal, since any linear $k$-uniform hypergraph on $n$
vertices has at most ${n \choose 2} / {k \choose 2}$ hyperedges: The
$n$ vertices provide us with ${n \choose 2}$ vertex pairs.  Each
hyperedge occupies ${k \choose 2}$ pairs, and because of linearity, no
pair can be occupied by more than one hyperedge.

\begin{proof}
  Choose the vertex set $V = V_1 \uplus \dots \uplus V_k$, where each
  $V_i$ is a disjoint copy of the finite field $GF(q)$. The hyperedges
  consist of all $k$-tuples $(x_1,\dots,x_k)$ with $x_i \in V_i, 1\leq
  i\leq k$, such that
  \begin{eqnarray}
  \left(
    \begin{array}{cccccc}
      1 & 1 &       &  1   &        & 1 \\
      1 & 2 & \dots &  i   &  \dots & k \\
      1 & 4 &       &  i^2 &        & k^2 \\      
      \vdots  & \vdots & & \vdots & & \vdots \\
      1 & 2^{k-3} & \dots & i^{k-3} & \dots & k^{k-3}\\
    \end{array}
  \right)
  \left(
    \begin{array}{c}
      x_1 \\
      x_2 \\
      \vdots\\
      x_i\\
      \vdots\\
      x_k
    \end{array}
  \right)      
   = \mathbf{0}  \ .
   \label{vdm}
 \end{eqnarray}

 Consider two distinct vertices $x \in V_i$, $y \in V_j$. How many
 hyperedges contain both of them? If $i=j$, none. If $i\ne j$, we can
 find out by plugging the fixed values $x, y$ into (\ref{vdm}). We
 obtain a (possibly non-uniform) $(k-2)\times(k-2)$ linear system with
 a Vandermonde matrix, which has a unique solution. In other words,
 $x$ and $y$ are in exactly one hyperedge, and the hypergraph is
 linear.  By the same argument, there are exactly $q^2$ hyperedges.
\end{proof}

\begin{proof}[Proof of the upper bound in Theorem~\ref{theorem-bounds}]
  Choose a prime power $q \in \{k2^k, \dots, 2k2^k-1\}$.  By
  Lemma~\ref{lemma-vdm}, there is a linear $k$-uniform hypergraph $H$
  with $n=qk$ vertices and $m=q^2$ hyperedges. Since $\frac{m}{n} =
  \frac{q}{k} \geq 2^k$, Lemma~\ref{random-formula} shows that there
  is an unsatisfiable linear $k$-CNF formula with $q^2 \leq 4k^24^k$
  clauses.
\end{proof}
  
Let us prove the lower bound of Theorem~\ref{theorem-bounds}.  For a
literal $u$ and a CNF formula $F$, we write $\occ_F(u) :=| \{C \in F \
| \ u \in C\} |$, the {\em degree} of the literal $u$. Thus $d_F(x)
=\occ_F(x)+\occ_F(\bar{x})$. We write $\occ(F) = \max_u \occ_F(u)$.
In analogy to $f(k)$, we define $\focc(k)$ to be the largest integer
$d$ such that any $k$-CNF formula $F$ with $\occ(F) \leq d$ is
satisfiable.  Clearly $\focc(k) \geq \frac{f(k)}{2}$, and thus
from~\cite{KST93} it follows that $\focc(k) \geq \frac{2^k}{2ek}$.
Actually, an application of the {\em Lopsided} Lov\'asz Local
Lemma~\cite{erdos91,AS00,LS07} yields $\focc(k) \geq \frac{2^k}{ek}-1$.

\begin{lemma}
  Let $F$ be a linear $k$-CNF formula with at most $1+ \focc(k-1)$
  variables of degree at least $1+ \focc(k-1)$. Then $F$ is
  satisfiable.
  \label{lemma-frequent}
\end{lemma}
\begin{proof}
  Transform $F$ into a $(k-1)$-CNF formula $F'$ by removing in every
  clause in $F$ a literal of maximum degree. We claim that
  $\deg_{F'}(u) \le \focc(k-1)$ for every literal $u$. Therefore $F'$
  is satisfiable, and $F$ is, as well.

  For the sake of contradiction, suppose there is a literal $u$ such
  that $t:= \occ_{F'}(u) \ge 1+\focc(k-1)$. Let $C'_i$,
  $i=1,2,\ldots,t$, be the clauses in $F'$ containing $u$. $C'_i$ is
  obtained by removing some literal $v_i$ from some clause $C_i \in
  F$. By construction of $F'$, $\occ_F(v_i) \ge \occ_F(u) \ge
  \focc(k-1)+1$ for all $1\leq i \leq t$. The $v_i$ are pairwise
  distinct: If $v_i = v_j$, then $\{u,v_i\} \subseteq C_i \cap C_j$.
  Since $F$ is weakly linear, this can only mean $i=j$. Now
  $u,v_1,v_2,\ldots,v_t$ are $t+1 \ge 2+\focc(k-1)$ variables of
  degree at least $1+\focc(k-1)$ in F, a contradiction.
\end{proof}

We see that an unsatisfiable weakly linear $k$-CNF formula has at
least $\focc(k-1)+2 \geq \frac{2^k}{2e(k-1)}+1$ literals of degree at
least $\focc(k-1)+1 \geq \frac{2^k}{2e(k-1)}$.  Double counting yields
$k|F| = \sum_u \occ_F(u) > \frac{4^k}{4e^2(k-1)^2}$, thus $|F| >
\frac{4^k}{16e^2 k^3}$.  By a more careful argument, we can improve
this by a factor of $k$.  We call a hypergraph {\em $(j,d)$-rich} if
at least $j$ vertices have degree at least $d$. The following lemma is
due to Welzl~\cite{emo-personal}.

\begin{lemma}
  For $d \in \NNnull$, every linear $(d,d)$-rich hypergraph has at
  least $\binom{d+1}{2}$ edges. This bound is tight for all $d \in
  \NNnull$.
\label{lemma-rich}
\end{lemma}
\begin{proof}

  We proceed by induction over $d$.  Clearly, the assertion of the
  lemma is true for $d=0$. Now let $H=(V,E)$ be a linear $(d,d)$-rich
  hypergraph for $d\ge1$. Choose some vertex $v$ of degree at least
  $d$ in $H$ and let $H'=(V,E')$ be the hypergraph with $E' := E
  \setminus \{e \in E \,|\, e \ni v \}$. We have (i) $|E| \ge |E'| +
  d$, (ii) $H'$ is linear, since this property is inherited when edges
  are removed, and (iii) $H'$ is $(d-1,d-1)$-rich, since for no vertex
  other than $v$ the degree decreases by more than $1$ due to the
  linearity of $H$. It follows hat $|E| \ge \binom{d}{2} + d =
  \binom{d+1}{2}$.  The complete $2$-uniform hypergraph (graph, so to
  say) on $d+1$ vertices shows that the bound given is tight for all
  $d \in \NNnull$.
\end{proof}

\begin{proof}[Proof of the lower bound in Theorem~\ref{theorem-bounds}]
  A weakly linear $k$-CNF formula $F$ is a linear $k$-uniform
  hypergraph, with literals as vertices.  If $F$ is unsatisfiable,
  then by Lemma~\ref{lemma-frequent}, it is
  $(\focc(k-1)+1,\focc(k-1)+1)$-rich.  By Lemma~\ref{lemma-rich}, $F$
  has at least ${\focc(k-1)+2 \choose 2} > \frac{4^k}{8e^2 (k-1)^2}$
  clauses.
\end{proof}

There is an obvious generalization of the notion of being linear.  We
say a CNF formula is {\em $b$-linear}, if any two distinct clauses
$C,D \in F$ fulfill $|\var(C) \cap \var(D)| \leq b$, and {\em weakly
  $b$-linear} if $|C \cap D| \leq b$ holds for all distinct $C,D \in
F$. Thus, a (weakly) $1$-linear formula is (weakly) linear.  We can
generalize Theorem~\ref{theorem-bounds} for $b \geq 2$.  However, the
proofs do not introduce new ideas and goes along 
the lines of the proofs presented above.

\begin{theorem}
  Let $b \geq 2$. Every weakly $b$-linear $k$-CNF formula with at most
  $\frac{2^{k(1 + \frac{1}{b})}}{2^{b+2}e^2k^{2+\frac{1}{b}}}$ clauses
  is satisfiable. There exists an unsatisfiable $b$-linear $k$-CNF
  formula with at most $2^{b+1} (k2^k)^{1 + \frac{1}{b}}$ clauses.
\label{theorem-b-linear}
\end{theorem}

\section{Proof of Theorem~\ref{theorem-treelike}}
\label{construction}

Let $F$ be an unsatisfiable weakly linear $k$-CNF formula, and let $T$
be a resolution tree of minimal size of $F$. We want to show that $T$
has a large number of nodes. It is not difficult to see that a
resolution tree of minimal size is {\em regular}, meaning that no
variable is resolved more than once on a path from a leaf to the root.
See Urquhart~\cite{Urquhart95}, Lemma 5.1, for a proof of this fact.
We take a random walk of length $\ell$ in $T$ starting at the root, in
every step choosing randomly to go to one of the two children of the
current node. If we arrive at a leaf, we stay there.  We claim that if
$\ell \leq \sqrt{2^{k-2}}$, then with probability at least
$\frac{1}{2}$, our walk does not end at a leaf. Thus, $T$ has at least
$2^{\ell-1}$ inner vertices at distance $\ell$ from the root, thus
at least $2^{2^{\frac{k}{2}-1}}$ leaves.\\

\begin{figure}
  \label{figure-tree}
  \begin{center}
    \includegraphics[width=0.4\textwidth]{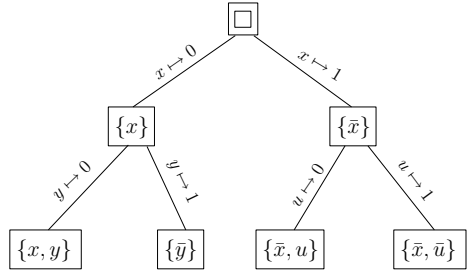}
    \caption{A resolution tree, with its 
      edges labeled in the obvious way. Every
      clause is unsatisfied when applying the assignments on the path
      to the root.}
  \end{center}
\end{figure}

As illustrated in Figure~1, we label each edge in $T$
with an assignment. If $C$ is the resolvent of $D_1$ and $D_2$, $x \in
D_1$ and $\bar{x} \in D_2$, we label the edge from $C$ to $D_1$ by $x
\mapsto 0$ and from $C$ to $D_2$ by $x \mapsto 1$. Each path from the
root to a node gives a partial assignment $\alpha$. If that node is
labeled with clause $C$, then $C$ evaluates to \texttt{false} under
$\alpha$. In our random walk, let $\alpha_i$ denote the partial
assignment associated with the first $i$ steps.  $\alpha_0$ is the
empty assignment, and $\alpha_i$ assigns exactly $i$ variables (if we
are not yet at a leaf). We set $F_i := F^{[\alpha_i]}$, i.e., the
formula obtained from $F$ by fixing the variables according to the
partial assignment $\alpha_i$. For a formula $G$, we define the {\em
  weight} $w(G)$ to be
\begin{eqnarray}
w(G) := \sum_{C \in G, |C| \leq k-2} 2^{k - |C|} \ .
\label{contrib}
\end{eqnarray}
Since $F$ is a $k$-CNF formula, $w(F) = 0$. If some formula $G$
contains the empty clause, then $w(G) \geq 2^k$. In our random walk,
$w(F_i)$ is a random variable.

\begin{lemma}
$\E[w(F_{i+1})] \leq \E[w(F_i)] + 4i$.
\end{lemma}

Since $w(F_0) = 0$, this implies $\E[w(F_\ell)] \leq 4{\ell \choose 2}
\leq 2\ell^2$. If our random walk ends at a leaf, then $F_\ell$
contains the empty clause, thus $w(F_\ell) \geq 2^k$.  Therefore
$2\ell^2 \geq \E[w(F_\ell)] \geq 2^k \Pr [ \textnormal {the random
  walk ends at a leaf}]$.  We conclude that at least half of all paths
of length $\ell^* = \sqrt{2^{k-2}}$ starting at the root do not end at
a leaf. Thus $T$ has at least $2^{\ell^*-1}$ internal nodes at
distance $\ell^*$ from the root, and thus at least $2^{\ell^*}$
leaves, which proves the theorem. It remains to prove the lemma.

\begin{proof}[Proof of the lemma]
  For a formula $G$ and a variable $x$, let $d_{k-1}(x,G)$ denote the
  number of $(k-1)$-clauses containing $x$ or $\bar{x}$. Since $F_0$
  is a $k$-CNF formula, $d_{k-1}(x,F_0) = 0$, for all variables $x$.
  We claim that $d_{k-1}(x,F_{i+1}) \leq d_{k-1}(x,F_i)+2$ for every
  variable $x$. To see this, note that in step $i$, some variable $y$
  is set to $b \in \{0,1\}$, say to $0$. At most one $k$-clause of
  $F_i$ contains $y$ and $x$, and at most one contains $y$ and
  $\bar{x}$, since $F_i$ is weakly linear, thus $d_{k-1}(x,F_{i+1})
  \leq d_{k-1}(x,F_i)+2$. It follows immediately
  that $d_{k-1}(x, F_i) \leq 2i$.\\

  Consider $w(F_i)$, which was in (\ref{contrib}).  $F_{i+1}$ is
  obtained from $F_i$ by setting some variable $y$ randomly to $0$ or
  $1$. Consider a clause $C$. How does its contribution to
  (\ref{contrib}) change when setting $y$? If (i) $y \not \in \var(C)$
  or $|C|=k$, it does not change. If (ii) $y \in \var(C)$ and $|C|
  \leq k-2$, then with probability $\frac{1}{2}$ each, its
  contribution to (\ref{contrib}) doubles or vanishes. Hence on
  expectation, it does not change. If (iii) $y \in \var(C)$ and $|C| =
  k-1$, then $C$ contributes nothing to $w(F_i)$, and with probability
  $\frac{1}{2}$, it contributes $4$ to $w(F_{i+1})$. On expectation,
  its contribution to (\ref{contrib}) increases by $2$.  Case (iii)
  applies to at most $d_{k-1}(y,F_i) \leq 2i$ clauses.  Hence
  $\E[w(F_{i+1})] \leq \E[w(F_i)] + 4i$.
\end{proof}

\section{Proof of Theorem~\ref{theorem-strictly-treelike}}
 
Let $F$ be a strictly treelike weakly linear $k$-CNF formula $F$, and
let $T$ be a strict resolution tree of $F$. Letters $a,b,c$ denote
nodes of $T$, and $u,v,w$ denote literals. Every node $a$ of $T$ is
labeled with a clause $C_a$. We define a graph $G_a$ with vertex set
$C_a$, connecting $u,v\in C_a$ if $u,v \in D$ for some clause $D \in
F$ that occurs as a label of a leaf in the subtree of $a$. Since $T$
is a strict resolution tree and $F$ is weakly linear, every edge in
$G_a$ comes from a unique leaf of $T$.  Resolution now has a simple
interpretation as a ''calculus on graphs'', see Figure~2. If $a$ is a
leaf, then $G_a = K_k$. Since the root of a resolution tree is labeled
with the empty clause, we have $G_{\rm root}=(\emptyset,\emptyset)$,
\ignore{the graph containing no vertices. We define a kind of
  ``complexity measure'' for nodes $a$ in $T$ in terms of $G_a$. It
  should be small for the root, large for the leaves, and decrease
  only slowly when moving from a leaf to the root. Our complexity
  measure will not be a single number, but a tuple of numbers.} For a
graph $G$, let $\kappa_i(G)$ denote the minimum size of a set $U
\subseteq V(G)$ such that $G-U$ contains no $i$-clique. Here, $G-U$ is
the subgraph of $G$ induced by $V(G)\setminus U$. Thus, $\kappa_1(G) =
|V(G)|$, and $\kappa_2(G)$ is the size of a minimum vertex cover of
$G$. For the complete graph $K_k$, $\kappa_i(K_k) = k-i+1$. We write
$\kappa_i(a) := \kappa_i(G_a)$. The tuple $(\kappa_1(a), \dots,
\kappa_k(a))$ can be viewed as the complexity measure for $a$.  We
observe that if $a$ is a leaf, then $\kappa_i(a)=k-i+1$, and
$\kappa_i({\rm root})=0$, for all $1\leq i \leq k$.  If $a$ is an
ancestor of $b$ in $T$, let $\dist(a,b)$ denote the number of edges in
the $T$-path from $a$ to $b$. Since one resolution step deletes one
literal (and may add several), the next proposition is immediate:

\begin{figure}
\begin{center}
\includegraphics[width=0.5\textwidth]{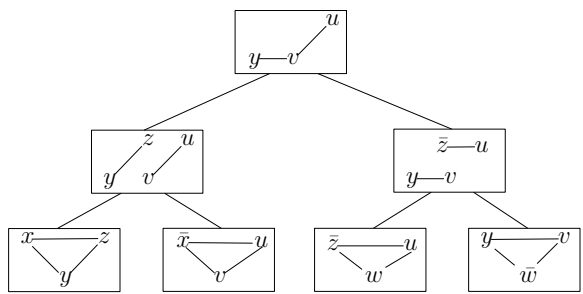}
\caption{Resolution as a calculus on graphs. A resolution step amounts
  to deleting the resolved vertex and taking the union of the two
  graphs.}
\end{center}
\label{figure-graph-resolution}
\end{figure}

\begin{proposition}
  If $b$ is a descendant of $a$ in $T$, then $\kappa_i(b) \leq
  \kappa_i(a)+\dist(a,b)$.
  \label{lipschitz}
\end{proposition}
\ignore{
\begin{proof}
  If $\dist(a,b)=0$, it is trivial. Suppose $\dist(a,b)=1$, i.e.  $b$
  is a child of $a$.  Let $C_a$ be the resolvent of $C_b$ and $C_c$,
  $u \in C_b$, $\bar{u} \in C_c$. Since $C_b \setminus \{u\} \subseteq
  C_a$, $G_b - \{u\}$ is a subgraph of $G_a$. If there is a set $U
  \subseteq V(G_a)$ such that $G_a - U$ does not have an $i$-clique,
  then surely $G_b - (U \cup \{u\})$ does not have an $i$-clique,
  either. Thus $\kappa_i(b) \leq \kappa_i(a)+1$. For $\dist(a,b)\geq
  2$, the claim follows immediately by induction.
\end{proof}
  }
At this point we want to give an intuition of the proofs that follow.
Our goal is to show that if the values $\kappa_i(a)$ are small for
some node $a$ in the tree, then the subtree of $a$ is big. The proof
goes roughly as follows: If the subtree of $a$ is small, then there
are many descendants $b$ of $a$ that are not too far from $a$ and have
even smaller subtrees. By induction, we will be able to show that
$\kappa_{i+1}(b)$ is fairly large. Thus, on the path from $b$ to $a$,
not all $(i+1)$-cliques are destroyed, and every such descendant $b$
of $a$ provides $G_a$ with an $(i+1)$-clique. These cliques need not
be vertex-disjoint, but they are edge-disjoint. This implies that
$G_a$ has many vertex-disjoint $i$-cliques, a contradiction to
$\kappa_i(a)$ being small. To make this intuition precise, we have to
define what small and big actually means in this context: We fix a
value $1 \leq \ell \leq k$ and define $\nu_i$ and $\theta_i$ for
$1\leq i \leq \ell$ as follows: $\theta_{\ell} := \left\lfloor
  \frac{k-\ell+1}{2}\right\rfloor -1$ and $\nu_{\ell} := 1$, and for
$1 \leq i < \ell$, we inductively define $\theta_i :=
\left\lfloor\frac{
    2^{\nu_{i+1}\theta_{i+1}-2}}{\theta_{i+1}}\right\rfloor -1$ and
$\nu_{i} := \frac{\nu_{i+1}\theta_{i+1}-1}{\theta_i}
\left\lfloor\frac{\theta_i}{\theta_{i+1}}\right\rfloor$. One should
not worry about these ugly expressions too much, they are only chosen
that way to make the induction go through.  \ignore{
\begin{eqnarray*}
  \theta_{\ell} & := & \left\lfloor \frac{k-\ell+1}{2} 
  \right\rfloor -1 \\

  \nu_{\ell} & := & 1\\
  & & \\
  \theta_i & := & \left\lfloor\frac{
      2^{\nu_{i+1}\theta_{i+1}-2}}{\theta_{i+1}}\right\rfloor -1 
\ , \quad 1 \leq i < \ell \\
  \nu_{i} & := & \frac{\nu_{i+1}\theta_{i+1}-1}{\theta_i}
  \left\lfloor\frac{\theta_i}{\theta_{i+1}}\right\rfloor 
  \ , \quad 1 \leq i < \ell \ .
\end{eqnarray*}
}
For the right value of $\ell$, one checks that $\theta_1$ is a tower
function in $k$.  More precisely, for any $\epsilon >0$, there exists
a $c \in \mathbb{N}$ such that when choosing $\ell = k - c$, then
$\theta_1 \geq {\rm tower}_{2-\epsilon}(k-c)$. The following theorem
is a more precise version of Theorem~\ref{theorem-strictly-treelike}.

\begin{theorem}
  Let $F$ be a strictly treelike linear $k$-CNF formula. Then $F$ has
  at least $2^{\nu_1\theta_1}$ clauses.
\end{theorem}

\begin{proof}
  A node $a$ in $T$ is {\em $i$-extendable} if $\kappa_j(a) \leq
  \theta_j$ for each $i \leq j \leq \ell$.  We observe that if $a$ is
  $i$-extendable, it is also $(i+1)$-extendable. For $i = \ell+1$, the
  condition is void, so every node is $(\ell+1)$-extendable. Also, the
  root is $1$-extendable, since $\kappa_1({\rm root})=0$.

  \begin{definition}
    A set $A$ of descendants of $a$ in $T$ such that (i) no vertex in
    $A$ is an ancestor of any other vertex in $A$ and (ii) $\dist(a,b)
    \leq d$ for all $b \in A$ is called an {\em antichain of $a$ at
      distance at most $d$}. If furthermore every $b\in A$ is
    $i$-extendable, we call $A$ an $i$-extendable antichain.
  \end{definition}

  \begin{lemma}
    Let $1\leq i \leq \ell$, and let $a$ be a node in $T$. If $a$ is
    $i$-extendable, then there is an $(i+1)$-extendable antichain $A$ of
    $a$ at distance at most $\theta_i$ such that $|A| =
    2^{\nu_i\theta_i}$.
    \label{lemma-extend}
  \end{lemma}

  \begin{proof}
    We use induction on $\ell -i$. For the base case $i=\ell$, we have
    $\kappa_{\ell}(a) \leq \theta_{\ell}$, as $a$ is
    $\ell$-extendable.  Since each leaf $b$ of $T$ has
    $\kappa_{\ell}(b) = k-\ell+1 \geq 2\theta_{\ell}+2$,
    Proposition~\ref{lipschitz} tells us that every leaf in the
    subtree of $a$ has distance at least $\theta_{\ell}+2$ from $a$.
    Since $T$ is a complete binary tree, there are $2^{\theta_\ell}$
    descendants of $a$ at distance exactly $\theta_\ell$ from $a$.
    This is the desired antichain $A$ of $a$.  Since every node is
    $(\ell+1)$-extendable, the base case holds.  For the step, let $a$
    be $i$-extendable, for $1 \leq i < \ell$.\\
    
    {\em Claim:} Let $b$ be a descendant of $a$ with $\dist(a,b) \leq
    \theta_i$.  If $b$ is $(i+1)$-extendable, then there is an
    $(i+1)$-extendable antichain $A$ of $b$ at distance at most
    $\theta_{i+1}$ of size $2^{\nu_{i+1}\theta_{i+1}-1}$.  
   \begin{figure}
    \begin{center}
    \includegraphics[width=0.3\textwidth]{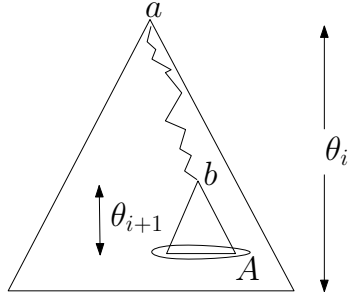}
    \caption{Illustration of the claim in the proof of
      Lemma~\ref{lemma-extend}. If node $a$ is $i$-extendable, and $b$
      is a close $(i+1)$-extendable descendant of $a$, then $b$
      itself has many close descendants $A$, at least half of
      which are $(i+1)$-extendable themselves.}
    \end{center}
    \label{figure-extendable}
    \end{figure}            

    \begin{proof}[Proof of the claim]
      By applying the induction hypothesis of the lemma to $b$, there
      is an $(i+2)$-extendable antichain $A$ of $b$ at distance at
      most $\theta_{i+1}$ of size $2^{\nu_{i+1}\theta_{i+1}}$.  Let
      $A_{\rm good} := \{c \in A \ | \ \kappa_{i+1}(c) \leq
      \theta_{i+1}\}$. This is an $(i+1)$-extendable antichain. If
      $A_{\rm good}$ contains at least half of $A$, we are done.  See
      Figure~3 for an illustration. Write $A_{\rm bad} := A \setminus
      A_{\rm good}$ and suppose for the sake of contradiction that
      $A_{\rm bad} > 2^{\nu_{i+1}\theta_{i+1}-1}$.  Consider any $c
      \in A_{\rm bad}$.  On the path from $c$ to $b$, in each step
      some literal gets removed (and others may be added). Let $P$
      denote the set of the removed literals. Then $C_c \setminus
      \{P\} \subseteq C_b$, and $G_c - P$ is a subgraph of $G_b$. Node
      $c$ is not $(i+1)$-extendable, thus $\kappa_{i+1}(c) \geq
      \theta_{i+1}+1$.  Since $|P| = \dist(b,c) \leq \theta_{i+1}$,
      the graph $G_c - P$ contains at least one $(i+1)$-clique, which
      is also contained in $G_b$.  This holds for every $c\in A_{\rm
        bad}$, and by weak linearity, $G_b$ contains at least $|A_{\rm
        bad}|$ edge disjoint $(i+1)$-cliques. Since $b$ is
      $(i+1)$-extendable, there exists a set $U \subseteq V(G_v),
      |U|=\kappa_{i+1}(b)$ such that $G_b - U$ contains no
      $(i+1)$-clique. Each of the $|A_{\rm bad}|$ edge-disjoint
      $(i+1)$-cliques $G_b$ contains some vertex of $U$, thus some
      vertex $v \in U$ is contained in at least $\frac{|A_{\rm
          bad}|}{|U|} \geq
      \frac{2^{\nu_{i+1}\theta_{i+1}-1}}{\theta_{i+1}} \geq 2\theta_i
      + 1$ edge-disjoint $(i+1)$-cliques. Two such cliques overlap in
      no vertex besides $v$, hence $G_b$ contains at least
      $2\theta_i+1$ {\em vertex-disjoint} $i$-cliques, thus
      $\kappa_i(b) \geq 2\theta_i+1$.  By Proposition~\ref{lipschitz},
      $\kappa_i(a) \geq \kappa_i(b) - \dist(a,b) \geq \theta_i + 1$.
      This contradicts the assumption of Lemma~\ref{lemma-extend} that
      $a$ is $i$-extendable. We conclude that $|A_{\rm bad}| \leq
      \frac{1}{2} |A|$, which proves the claim.
    \end{proof}

    Let us continue with the proof of the lemma. If $A$ is an
    $(i+1)$-extendable antichain of $a$ at distance $d \leq \theta_i$,
    then by the claim for each vertex $b \in A$ there exists an
    $(i+1)$-extendable antichain of $b$ at distance at most
    $\theta_{i+1}$, of size $2^{\nu_{i+1}\theta_{i+1}-1}$.  Their
    union is an $(i+1)$-extendable antichain $A'$ of $a$ at distance
    at most $d + \theta_{i+1}$, of size
    $|A|2^{\nu_{i+1}\theta_{i+1}-1}$.  Hence we can ``inflate'' $A$ to
    $A'$, as long as $d \leq \theta_i$.  Starting with the
    $(i+1)$-extendable antichain $\{a\}$ and inflate it
    $\left\lfloor\frac{\theta_{i}}{\theta_{i+1}}\right\rfloor$ times,
    and obtain a final $(i+1)$-extendable antichain of $a$ at distance at
    most $\theta_i$ of size at least
    $\left(2^{\nu_{i+1}\theta_{i+1}-1}\right)^{\left
        \lfloor\frac{\theta_{i}}{\theta_{i+1}}\right\rfloor} =
    2^{\nu_i\theta_i} $.
  \end{proof}
  
  Applying Lemma~\ref{lemma-extend} to the root of $T$, which is
  $1$-extendable, we obtain an antichain $A$ of size
  $2^{\nu_1\theta_1}$ nodes. Since $T$ has at least $|A|$ leaves, this
  proves the theorem.
\end{proof}

\section{Open Problems}

Let $\flin(k)$ be the largest integer $d$ such that any linear $k$-CNF
formula $F$ with $d(F) \leq d$ is satisfiable. Clearly $\flin(k) \geq
f(k)$, and from the proof of the upper bound in
Theorem~\ref{theorem-bounds} it follows that $\flin(k) \leq 2k2^k$. Is
there a significant gap between $f(k)$ and $\flin(k)$? It is not
difficult to show that $f(2)=\flin(2)=2$, but we do not know the value
of $\flin(k)$ for any $k \geq 3$.
\ignore{ The bound $\flin(2) \geq 2$ follows from
$f(k) \geq k$, and $\flin(2)\leq 2$ is witnessed by
$$
(\bar{x}_1 \vee x_2) \wedge (\bar{x}_2 \vee x_3) \wedge
(\bar{x}_3 \vee x_4) \wedge (\bar{x}_4 \vee x_1) \wedge
(x_1 \vee x_3) \wedge (\bar{x}_2 \vee \bar{x}_4)  \ .
$$
}
How do unsatisfiable linear $k$-CNF formulas look like? Can one find
an explicit construction of an unsatisfiable linear $k$-CNF formula
whose size is singly exponential in $k$? We suspect one has to come up
with some algebraic construction.  What is the resolution complexity
of linear $k$-CNF formulas?  Tree resolution complexity is doubly
exponential in $k$.  We suspect the same to be true for general
resolution.

\section*{Acknowledgments}

My thanks go to Emo Welzl, Robin Moser, Heidi Gebauer, Andreas Razen
and Philipp Zumstein for very helpful and fruitful discussions.

\bibliographystyle{abbrv}

\bibliography{scheder}

\newpage
\appendix

\section{Proof of Theorem~\ref{theorem-b-linear}}

To prove the lower bound, we need an analog of
Lemma~\ref{lemma-frequent}.

\begin{lemma}
  Let $F$ be a weakly $b$-linear $k$-CNF formula. Denote by $\ell$ the
  number of literals occurring in at least $\focc(k-b) +1$ clauses.
  If ${\ell-1 \choose b} \leq \focc(k-b)$, then $F$ is satisfiable.
\label{b-frequent}
\end{lemma}
\begin{proof}
  Transform $F$ into a $(k-b)$-CNF formula $F'$ by removing in every
  clause in $F$ some $b$ literals of maximum degree, breaking ties
  arbitrarily.  We claim that $\occ_{F'}(u) \le \focc(k-b)$ for all
  literals $u$.  Since $F'$ is a $(k-b)$-CNF formula, this means that
  $F'$ is satisfiable, thus $F$ is, too.

  Suppose $u$ is a literal with $t := \occ_{F'}(u) \ge 1+\focc(k-b)$;
  this implies $\occ_F(u) \ge 1 + \focc(k-b)$. Let $C'_i$,
  $i=1,2,\ldots,t$, be the clauses in $F'$ with $C'_i \ni u$.  Let
  $v^{(i)}_1,\dots,v^{(i)}_b$ be the $b$ literals that were removed so
  that $C'_i$ was obtained. All these literals $v^{(i)}_j$ fulfill
  $\occ_F(v^{(i)}_j) \geq \occ_F(u) \geq 1+\focc(k-b)$. Observe that
  because $F$ is weakly $b$-linear, the sets
  $\{v^{(i)}_1,\dots,v^{(i)}_b\}$, $1 \leq i \leq t$ are $t$ distinct
  sets of cardinality $b$ not containing $u$. Therefore $t \leq
  {\ell-1 \choose b} \leq \focc(k-b)$. This is a contradiction since
  we assumed $t \ge 1+\focc(k-b)$.
\end{proof}

If $F$ is an unsatisfiable weakly $b$-linear $k$-CNF formula, then
by Lemma~\ref{b-frequent}, the number $\ell$ of literals
$u$ with $\occ_F(u) \geq \focc(k-b)+1$ fulfills
${\ell-1 \choose b} \geq \focc(k-b)+1$. Thus,
$$
\frac{2^{k-b}}{2ek} \leq \focc(k-b)+1 \leq {\ell-1\choose b}
\leq \frac{(\ell-1)^b}{b!} \ ,
$$
and, since $b! \geq 2^{b-1}$ for $b \geq 2$, we obtain $\ell \geq
\sqrt[b]{\frac{2^k}{4ek}} + 1$. $F$ contains at most $\ell$ many
literals of degree at least $\focc(k-b)+1$, therefore
$$
|F| > \frac{2^{k(1 + \frac{1}{b})}}{2^{b+2}e^2k^{2+\frac{1}{b}}} \ .
$$
This is the lower bound of Theorem~\ref{theorem-b-linear}.\\

For an upper bound, we construct a $b$-linear hypergraph $H$ with $n$
vertices and $m$ hyperedges such that $\frac{m}{n} \geq 2^k$.
Lemma~\ref{random-formula} is easily seen extend to $b$-linear
formulas and hypergraphs and yield an unsatisfiable formula with $m$
clauses. We need a generalization of Lemma~\ref{lemma-vdm}.

\begin{lemma}
  For any prime power $q$, any $k \in \mathbb{N}$, and $b \in
  \{1,\dots,k\}$, there exists a $k$-uniform $b$-linear hypergraph
  with $kq$ vertices and $q^{1 + b}$ edges.
  \label{lemma-b-vdm}
\end{lemma}

\begin{proof}
  Choose the vertex set $V = V_1 \uplus \dots \uplus V_k$, where each
  $V_i$ is a disjoint copy of the finite field $GF(q)$. The hyperedges
  consist of all $k$-tuples $(x_1,\dots,x_k)$ with $x_i \in V_i, 1\leq
  i\leq k$, such that
  \begin{eqnarray}
  \left(
    \begin{array}{cccccc}
      1 & 1 &       &  1   &        & 1 \\
      1 & 2 & \dots &  i   &  \dots & k \\
      1 & 4 &       &  i^2 &        & k^2 \\      
      \vdots  & \vdots & & \vdots & & \vdots \\
      1 & 2^{k-b-2} & \dots & i^{k-b-2} & \dots & k^{k-b-2}\\
    \end{array}
  \right)
  \left(
    \begin{array}{c}
      x_1 \\
      x_2 \\
      \vdots\\
      x_i\\
      \vdots\\
      x_k
    \end{array}
  \right)      
   = \mathbf{0}  \ .
   \label{vdm-b}
 \end{eqnarray}

 The hyperedges form a $b$-linear hypergraph. To see this, consider
 $b+1$ distinct vertices $x_0 \in V_{i_0}, x_1 \in V_{i_1}, \dots, x_b
 \in V_{i_b}$.  How many hyperedges contain all of these $b+1$
 vertices? If the indices $i_0,\dots,i_b$ are not distinct, none does.
 Otherwise, we plug in the fixed values $x_0, \dots, x_b$ into
 (\ref{vdm-b}). We obtain a (possibly non-uniform)
 $(k-b-1)\times(k-b-1)$ linear system with a Vandermonde matrix, which
 has a unique solution. Hence those $b+1$ vertices are in exactly one
 common hyperedge.  By the same argument, there are exactly $q^{1+b}$
 hyperedges.
\end{proof}

We choose a prime power $q$ such that $\sqrt[b]{k2^k} \leq q <
2\sqrt[b]{k2^k}$. By Lemma~\ref{lemma-b-vdm}, there is a $b$-linear
hypergraph with $n=kq$ vertices and $m=q^{1+b}$ hyperedges. Then
$\frac{m}{n} = \frac{q^b}{k} \geq 2^k$, and by
Lemma~\ref{random-formula}, there is an unsatisfiable $b$-linear
hypergraph with at most $m \leq 2^{b+1} (k2^k)^{1+\frac{1}{b}}$. This
proves Theorem~\ref{theorem-b-linear}.

\end{document}